\documentclass[%
 reprint,
 amsmath,amssymb,
 aps,
]{revtex4-2}

\usepackage{graphicx}
\usepackage{dcolumn}
\usepackage{bm}
\usepackage{subfigure}
\usepackage{amssymb,mathtools}
\usepackage{latexsym}
\usepackage{bigints}
\usepackage{amsthm}
\newtheorem*{theorem}{Theorem}

\begin{document}

\preprint{APS/123-QED}

\title{Characteristics, Implementation and Applications of Special Perfect Entangler Circuits}

\author{M. Karthick Selvan}
\email{karthick.selvan@yahoo.com}%

\author{S. Balakrishnan}%
 \email{physicsbalki@gmail.com}
\affiliation{Department of Physics, School of Advanced Sciences, Vellore Institute of Technology, Vellore - 632014, Tamilnadu, India.}%



\begin{abstract}
We discuss the characteristics of special perfect entanglers and construct single parameter two-qubit circuits which are locally equivalent to special perfect entanglers. We present the results obtained from the implementation of one of the circuits using cross-resonance interaction and discuss their applications. First, we show that the ability of two-qubit gates to create entangled states can be described using the chords present in the argand diagram of squared eigenvalues of nonlocal part of two-qubit gates. We show that the entangling power of a two-qubit gate is proportional to the mean squared length of the chords. We deduce the entangling characteristics of special perfect entanglers from the argand diagram associated with them. We implement a special perfect entangler circuit using echoed cross-resonance gate and pulse-level programming for nine different circuit parameters. For a particular input state, we perform quantum state tomography and calculate state fidelity and concurrence of the obtained output density matrices. We also measure the average gate fidelity for B gate circuit. We construct two universal two-qubit quantum circuits using the special perfect entangler circuits. These universal circuits can be used to generate all two-qubit gates. We show that $(n-1)$ B gate circuits can be used to generate $n$-qubit GHZ and perfect W states. We generate three-qubit perfect W state. Perfect W state generated using pulse-level programming shows better fidelity than the state generated using echoed cross-resonance gate.
\end{abstract}

\maketitle 

\section{Introduction} 
Universality is an important aspect of entangling two-qubit gates in the circuit model of quantum computation~\cite{DiVincenzo1995,Barenco1995r,Deutsch1995,Lloyd1995,Barenco1995}. It implies the ability of an entangling two-qubit gate to simulate an arbitrary $n$-qubit unitary operation by multiple applications of the entangling two-qubit gate along with suitable single-qubit gates. Universality can be attributed to the nonlocal characteristics of two-qubit gates
such as their ability to create entanglement~\cite{Nielsen2003}. Usage of an entangling two-qubit gate as universal two-qubit gate relies on the maximum number of its applications required to generate all other two-qubit gates.  

Special perfect entanglers (SPEs) are entangling two-qubit gates with maximum entangling power~\cite{Rezakhani2004} which is defined as the average entanglement generated over all product states~\cite{Zanardi2000}. SPEs can transform an orthonormal product basis into an orthonormal maximally entangled basis~\cite{Rezakhani2004}. Three applications of SPEs are required to generate all two-qubit gates~\cite{Zhang2005}. B gate is an exemption in SPEs. It can generate all two-qubit gates in two applications~\cite{Zhang2004}. Among all the SPEs, only the gates belonging to controlled-NOT (CNOT) equivalence class are used as universal two-qubit gates for doing quantum computation in quantum hardwares~\cite{Debnath2016,ADC2013,Sheldon2016,Malekakhlagh2020}.

In this paper, we discuss the characteristics of argand diagram of squared eigenvalues of nonlocal part of SPEs. This argand diagram is used to explain the condition for perfect entanglers~\cite{Makhlin2002,Zhang2003}. The argand diagram of eigenvalues of unitray operators is also used to provide the conditions for discriminating two unitary operators~\cite{Acin2001,Duan2007,Chefles2005} and creation of perfect entanglers~\cite{Yu2010}. Those conditions are provided in terms of arc length~\cite{Acin2001,Duan2007,Yu2010}. The quantity that we consider, in this paper, is the chord length. The squared eigenvalues of nonlocal part of a two-qubit gate exist on a unit circle in the complex plane and form the four vertices of a quadrilateral. These four vertices are pairwise connected by six chords.   

In the case of symmetric two-qubit states, the concurrence is shown to be proportional to the square of chordal distance between two Majorana stars~\cite{Liu2016,Galindo2022}. We show that the length of a specific chord in the argand diagram of squared eigenvalues of nonlocal part of a two-qubit gate is proportional to the amount of entanglement generated by that two-qubit when it acts on a pair of orthonormal product states constructed using a pair of maximally entangled states as shown in Ref.~\cite{Kraus2001}. We show that the entangling power of a two-qubit gate is proportional to the average of squared length of six chords. In addition, we prove a condition for the two-qubit gates, which can also be inferred from Ref.~\cite{Kraus2001}, that can transform a pair of orthonormal product states into a pair of orthonormal maximally entangled states. 

The argand diagrams of squared eigenvalues of nonlocal part of SPEs have a common feature. There exist two specific chords passing through the origin for all SPEs. From this feature, we provide the exact orthonormal product basis which can be transformed into maximally entangled basis by all SPEs. For B gate, the quadrilateral formed by the squared eigenvalues is square in shape and hence it spans maximum area of the unit circle. By generalizing the B gate circuit given in Ref.~\cite{Zhang2004}, we construct four single parameter circuits and show that all of them are locally equivalent to SPEs. These circuits can be easily implemented using cross-resonance interaction where the control qubit is excited at the resonance frequency of target qubit~\cite{Rigetti2010,Chow2011}.  

In addition, we implement one of the SPE circuits in IBM quantum processor. Echoed cross-resonance (ECR) gate~\cite{ADC2013,Sheldon2016,Malekakhlagh2020} is used as basis gate in IBM quantum processors. Any two-qubit or multi-qubit gate can be implemented using ECR gate and suitable single-qubit gates. Alternatively, two-qubit gates can also be implemented in IBM quantum processors by pulse-level programming~\cite{Alexander2020}. This allows the rescaling of cross-resonance pulse area and reduce the implementation time~\cite{Stenger2021,Earnest2021,Satoh2022,Chen2022,Zhang2024}. We show the implementation of SPE circuit by both methods for a set of parameters in IBM quantum processor. In the pulse-level programming, we only rescale the area of cross-resonance pulses required to implement a part of SPE circuits. For a particular input state, we measure the output density matrices of the SPE circuit by performing quantum state tomography (QST) and calculate the state fidelity and concurrence. We also measure the average gate fidelity for the B gate circuit. Both implementations of SPE circuit, using ECR gates and pulse-level programming, provide nearly the same results. 

Universal two-qubit quantum circuit (UTQQC) is a parameterized two-qubit circuit that can be used to generate all two-qubit gates by changing the circuit parameters~\cite{Fan2005}. Universal two-qubit quantum circuits (UTQQCs) can be constructed using SPEs~\cite{Selvan2023}. We construct two UTQQCs using SPE circuits. These UTQQCs can be implemented in IBM quantum processors using cross-resonance pulses to generate any two-qubit gate operation. We also construct $n$-qubit circuits using $(n-1)$ B gate circuits to generate $n$-qubit GHZ state~\cite{Greenberger1990} and $n$-qubit perfect W-state~\cite{Agrawal2006}. Both states can be used for teleportation and superdense coding~\cite{Agrawal2006,Karlsson1998,Hao2001}. We generate three-qubit perfect W state using ECR gate and pulse-level programming and perform QST experiment. Better fidelity is obtained in the case where the three-qubit perfect W state is generated by pulse-level programming.
  
This paper is organized as follows. In section II, we uncover the relation between entangling power of two-qubit gates and chord lengths in the argand diagram of squared eigenvalues of nonlocal part of two-qubits. In section III, we discuss the characteristics of SPEs and introduce the SPE circuits. In section IV, we describe the implementation of SPE circuits and present the experimental results. In section V, we discuss the construction of UTQQCs and generation of genuinely entangled multi-qubit states as applications of SPE circuits. In section VI, we provide the conclusion. 

\section{A Revisit to Entanglement Creation by Two-Qubit Gates} 

Any two-qubit gate $(U)$ can be written in the following form~\cite{Zhang2003}.

\begin{equation}\label{eq1}
U = e^{i\alpha}k_1 U_d(c_1, c_2, c_3) k_2,
\end{equation}

where $k_1, k_2 \in SU(2) \otimes SU(2)$ are the local parts of $U$ and $U_d (c_1, c_2, c_3)$ is the nonlocal part of $U$ which can be expressed as 

\begin{equation*}
U_d (c_1, c_2, c_3) = \exp \bigg[ \dfrac{i H(c_1, c_2, c_3)}{2} \bigg]
\end{equation*}   
with 
\begin{equation}\label{eq2}
H(c_1, c_2, c_3) = c_1(\sigma_x \otimes \sigma_x) + c_2 (\sigma_y \otimes \sigma_y) + c_3 (\sigma_z \otimes \sigma_z)
\end{equation}   

The co-efficients $(c_1, c_2, c_3)$ are called Cartan co-ordinates. If the two-qubit gates, $U_1$ and $U_2$ are such that $U_1 = k_1 U_2 k_2$ for some $k_1, k_2 \in SU(2) \otimes SU(2)$ then they are said to be locally equivalent and it is denoted as $U_1 \sim U_2$. The set of all locally equivalent two-qubit gates form a local equivalence class. The local invariants can be used to distinguish gates belonging to different local equivalence classes~\cite{Makhlin2002,Zhang2003}. They are defined as 

\begin{equation}\label{eq3}
G_1 = \dfrac{tr^2[m(U)]}{16 \det{U}},
\end{equation}
\begin{equation}\label{eq4}
G_2 = \dfrac{tr^2[m(U)] - tr[m(U)^2]}{4 \det{U}},
\end{equation}

where $m(U)= (Q^{\dagger}U Q)^T(Q^{\dagger} U Q)$ with 

\begin{equation}\label{eq5}
Q = \dfrac{1}{\sqrt{2}}
\begin{bmatrix}
1 & 0 & 0 & i \\ 0 & i & 1 & 0 \\ 0 & i & -1 & 0 \\ 1 & 0 & 0 & -i
\end{bmatrix}.
\end{equation}

\begin{figure}[h]
\includegraphics[scale=0.55]{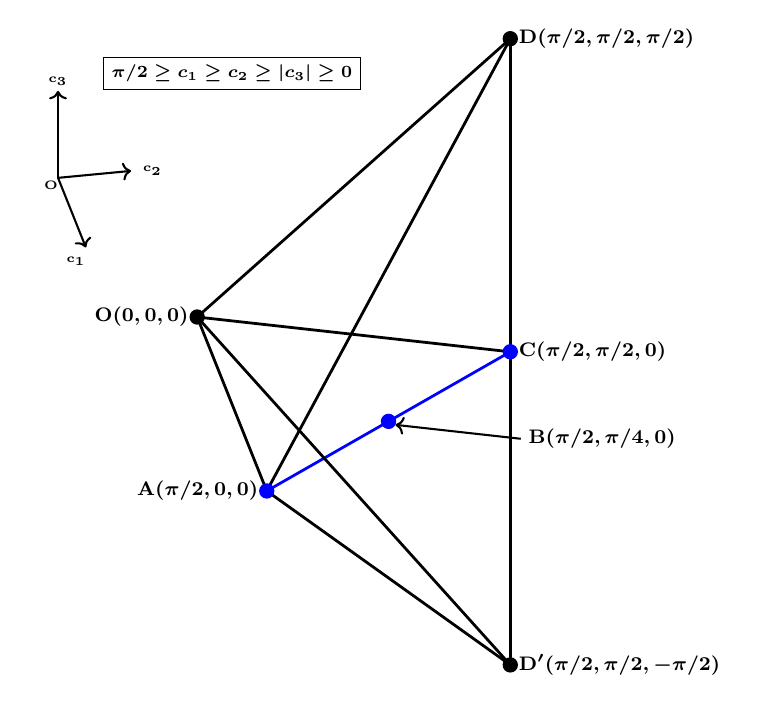}
\caption{Geometrical representation of local equivalence classes of two-qubit gates. Each point of the tetrahedron satisfy the condition: $\pi/2 \geq c_1 \geq c_2 \geq \vert c_3 \vert \geq 0$ and represent a local equivalence class. The blue line connecting the points $\text{A}(\pi/2, 0, 0)$ and $\text{C}(\pi/2, \pi/2, 0)$ represent SPEs. The midpoint of this line, $\text{B}(\pi/2, \pi/4, 0)$, represent B gate equivalence class.}
\label{fig1}
\end{figure}

The set of all local equivalence classes of two-qubit gates can be represented by the points of a tetrahedron (also known as Weyl chamber of two-qubit gates~\cite{Zhang2003}) as shown in FIG.~\ref{fig1}. The column vectors of the matrix $Q$ [Eq.~\ref{eq5}] form magic basis: $\{ \vert \Psi_1 \rangle = \frac{1}{\sqrt{2}} \left[ \vert 00 \rangle + \vert 11 \rangle \right], \vert \Psi_2 \rangle = \frac{i}{\sqrt{2}} \left[ \vert 01 \rangle + \vert 10 \rangle \right], \vert \Psi_3 \rangle = \frac{1}{\sqrt{2}} \left[ \vert 01 \rangle - \vert 10 \rangle \right], \vert \Psi_4 \rangle = \frac{i}{\sqrt{2}} \left[ \vert 00 \rangle - \vert 11 \rangle \right] \}$. They are the eigenbasis of $U_d(c_1, c_2, c_3)$ corresponding to the eigenvalues: $\{e^{ih_1/2}, e^{ih_2/2}, e^{ih_3/2}, e^{ih_4/2} \}$ with $h_j$'s being the eigenvalues of $H(c_1, c_2, c_3)$ given as follows. 

\[
h_1 = c_1 -c_2 + c_3, 
\]
\[
h_2 = c_1 + c_2 - c_3,
\]
\[
h_3 = -c_1-c_2-c_3,
\]
\begin{equation}\label{eq6}
h_4 = -c_1 + c_2 + c_3.
\end{equation} 

In terms of magic basis, the nonlocal part, $U_d(c_1, c_2, c_3)$, of a two-qubit gate can be expressed as follows. 

\begin{equation}\label{eq7}
U_d(c_1, c_2, c_3) =  e^{ih_3/2} \left[ \sum_{j=1}^4 e^{i(h_j-h_3)/2} \vert \Psi_j \rangle \langle \Psi_j \vert \right]. 
\end{equation}

Now we consider a product state $\vert \psi \rangle$ which can be expressed in the magic basis as $\vert \psi \rangle = \sum_{j=1}^4 \phi_j \vert \Psi_j \rangle$. Using the entanglement measure described in Ref.~\cite{Zhang2003}, the amount of entanglement created by $U_d$ upon acting on the product state $\vert \psi \rangle$ can be written as follows.

\begin{equation}\label{eq8}
\left| E \left( U_d \vert \psi \rangle \right)\right| = \dfrac{1}{2} \left| \sum_{j=1}^4 \phi_j^2 e^{i(h_j-h_3)} \right|.
\end{equation}

The co-efficients $\phi_j$'s satisfy the conditions~\cite{Zhang2003}: $\sum_j \phi_j^2 = 0$ and $\sum_j \vert \phi_j \vert^2 = 1$. When $U_d$ transforms the product state $\vert \psi \rangle$ into a maximally entangled state, the following relation is satisfied. 

\begin{equation}\label{eq9}
\sum_{j=1}^4 \vert \phi_j \vert^2 e^{i(h_j-h_3)} = 0. 
\end{equation}

This condition implies that the convex hull of $e^{i(h_j - h_3)}$'s should contain zero as shown in FIG.~\ref{fig2}; $e^{i(h_j - h_3)}$'s lie on a unit circle and there are six chords connecting them. 

\begin{figure}[h]
\includegraphics[scale=0.55]{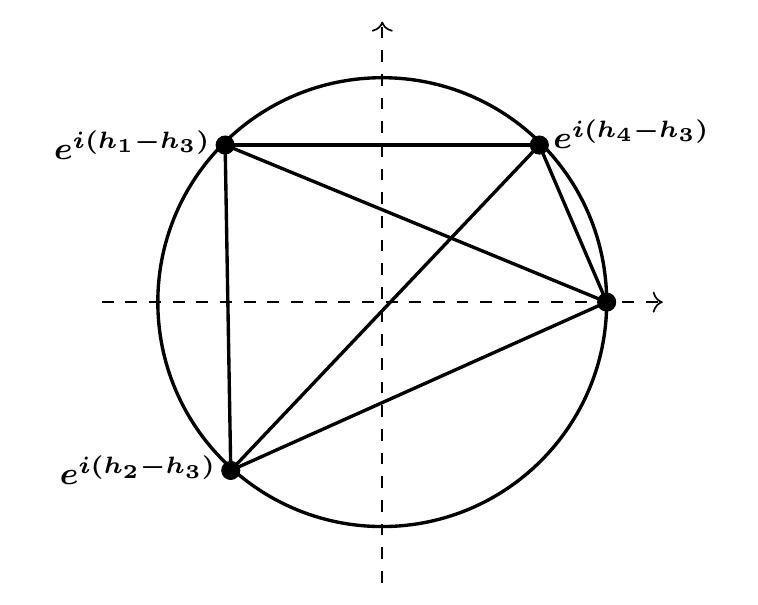}
\caption{Illustration of the condition for transforming a product state into a maximally entangled state.}
\label{fig2}
\end{figure}

\begin{theorem}
If the difference between two eigenvalues of $H(c_1, c_2, c_3)$ is such that $h_j - h_k = \pi$ then there exists a pair of orthonormal product states that can be transformed into a pair of orthonormal maximally entangled states by $U_d(c_1, c_2, c_3)$. 
\end{theorem}

\begin{proof}
Let $h_j$ and $h_k$ be two eigenvalues of $H(c_1, c_2, c_3)$ such that $h_j-h_k= \pi$. Then, in FIG.~\ref{fig2}, the chord connecting $e^{i(h_j-h_3)}$ and $e^{i(h_k-h_3)}$ passes through origin. This implies that there exists product states $\phi_j \vert \Psi_j \rangle + \phi_k \vert \Psi_k \rangle$ with $\phi_j$ and $\phi_k$ satisfying the conditions $\vert \phi_j \vert^2 e^{i(h_j-h_3)} + \vert \phi_k \vert^2 e^{i(h_k - h_3)} = 0$, $\vert \phi_j \vert^2 + \vert \phi_k \vert^2 = 1$, and $\phi_j^2 + \phi_k^2 = 0$. From these conditions, those product states can be written as $e^{i \theta} \left[ \dfrac{\vert \Psi_j \rangle \pm i \vert \Psi_k \rangle}{\sqrt{2}} \right]$. Ignoring the global phase, the states $\left[ \dfrac{\vert \Psi_j \rangle \pm i \vert \Psi_k \rangle}{\sqrt{2}} \right]$ are orthonormal product states and they are transformed into $e^{ih_j/2} \left[ \dfrac{\vert \Psi_j \rangle \pm \vert \Psi_k \rangle}{\sqrt{2}} \right]$ by $U_d(c_1, c_2, c_3)$. It can be verified that the value of $\vert E \vert$ is 0.5 for the transformed states and hence they are maximally entangled. 
\end{proof}

The length of the chords in FIG.~\ref{fig2} indicate the ability of $U_d$ to create entanglement. The amount of entanglement, $ \vert E(U_d \vert \psi^\pm_{jk} \rangle ) \vert$, created by $U_d$ acting upon the product states, $\vert \psi^\pm_{jk} \rangle = \left[ \dfrac{\vert \Psi_j \rangle \pm i \vert \Psi_k \rangle}{\sqrt{2}} \right]$, are the same and is proportional to the length of the chord $(L_{jk})$ connecting $e^{i(h_j - h_3)}$ and $e^{i(h_k - h_3)}$ in FIG.~\ref{fig2}. Specifically we have the following relation.

\begin{equation}\label{eq10a}
\vert E(U_d\vert \psi^\pm_{jk} \rangle ) \vert^2 = \dfrac{2 - 2 \cos(h_j - h_k)}{16} = \dfrac{L^2_{jk}}{16}.
\end{equation}

Taking average over the indices $jk$ on both sides, we get  

\begin{equation}\label{eq10b}
\overline{L_{jk}^2} = 16 \times \overline{\vert E(U_d\vert \psi^\pm_{jk} \rangle ) \vert^2}.
\end{equation}

In terms of Cartan co-ordinates, the expression of $\overline{L_{jk}^2}$ can be written as follows. 

\begin{equation}\label{eq10}
\overline{L_{jk}^2} = \dfrac{2}{3} \left[ 3 - \dfrac{1}{2}\left( \sum_{{p,q=1}; {p \neq q}}^3 \cos(2c_p) \cos(2c_q) \right) \right]
\end{equation}

The expression of entangling power $(e_p)$ of $U_d(c_1, c_2, c_3)$ in terms of Cartan co-ordinates can be written as follows~\cite{Rezakhani2004}. 

\begin{equation}\label{eq11}
e_p = \dfrac{1}{18} \left[ 3 - \dfrac{1}{2} \left( \sum_{{p,q=1}; {p \neq q}}^3 \cos(2c_p) \cos(2c_q) \right) \right]
\end{equation}

From Eqs.~\ref{eq10} and \ref{eq11}, the following relation can be obtained. 

\begin{equation}\label{eq12}
e_p = \dfrac{1}{12} \times \overline{L^2_{jk}}
\end{equation}

Thus the entangling power of any two-qubit gate with nonlocal part $U_d(c_1, c_2, c_3)$ is proportional to the average of squared length of chords in the argand diagram of squared eigenvalues of $U_d(c_1, c_2, c_3)$. Using Eq.~\ref{eq10b}, the above statement can be restated as the entangling power of any two-qubit gate with nonlocal part $U_d(c_1, c_2, c_3)$ is proportional to the average of squared value of entanglement measure (defined in Ref.~\cite{Zhang2003}) of the states resulted from the action of $U_d(c_1, c_2, c_3)$ on the product states $\vert \psi^\pm_{jk} \rangle$. 

\section{Characteristics of Special Perfect Entanglers and SPE Circuits}

\subsection{Characteristics of SPEs}

Special perfect entanglers (SPEs) are two-qubit gates with maximum entangling power~\cite{Rezakhani2004}. Their Cartan co-ordinates are $(\pi/2, c_2, 0)$ with $c_2 \in [0, \pi/2]$. The local invariants of SPEs are $G_1 = 0$ and $G_2 = cos(2c_2)$. In FIG.~\ref{fig1}, they are represented by the blue line connecting the points $\text{A}(\pi/2, 0, 0)$ and $\text{C}(\pi/2, \pi/2, 0)$. The points $\text{A}(\pi/2, 0, 0)$ and $\text{C}(\pi/2, \pi/2, 0)$ represent CNOT and Double-CNOT (or iSWAP) equivalence classes respectively. The midpoint $\text{B}(\pi/2, \pi/4, 0)$ represent B gate equivalence class. The nonlocal part of SPEs is of the following form. 
\begin{equation}\label{eq13}
U_d \left( \dfrac{\pi}{2}, c_2, 0 \right) = \exp \left[ \dfrac{i \pi}{4} (\sigma_x \otimes \sigma_x) + \dfrac{ic_2}{2} (\sigma_y \otimes \sigma_y) \right]
\end{equation}

\begin{figure}[h]
\begin{tabular}{cc}
\includegraphics[scale=0.35]{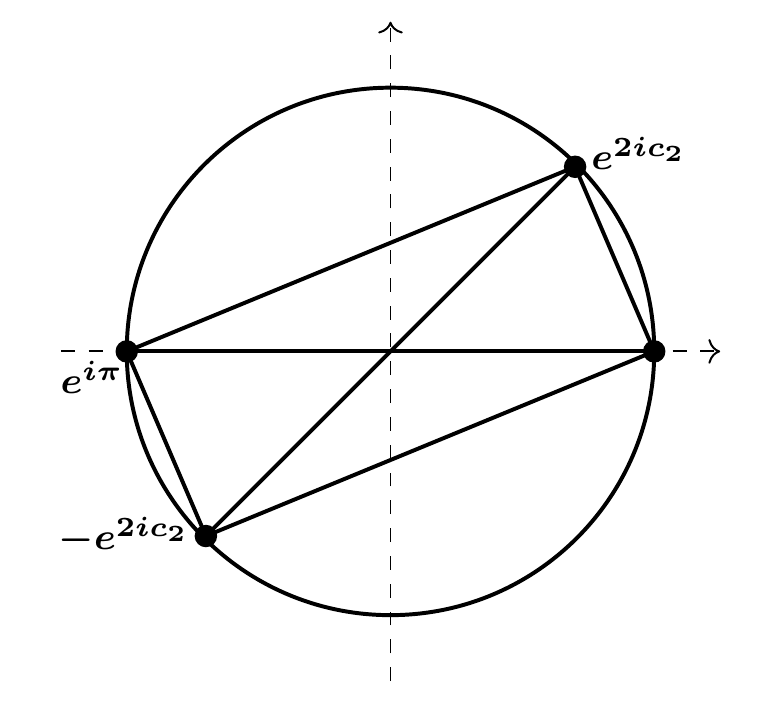} & \includegraphics[scale=0.35]{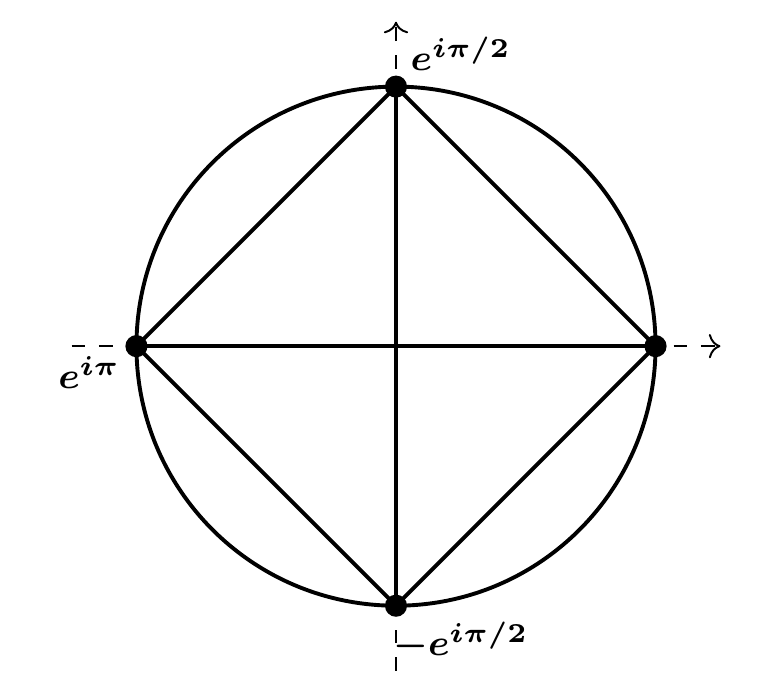} \\ (a) & (b)
\end{tabular}
\caption{Argand diagram of squared eigenvalues of nonlocal part of (a) a general SPE and (b) B gate.}
\label{fig3}
\end{figure}

The squared eigenvalues of $U_d(\pi/2, c_2, 0)$ form a unique structure in the complex plane as shown in FIG.~\ref{fig3}. The four points consist of two pairs of diametrically opposite points. Hence there exist two chords passing through origin. Consequently all SPEs transform the orthonormal product basis $\left\{ \vert \psi^\pm_{13} \rangle, \vert \psi^\pm_{24} \rangle \right\}$ into maximally entangled basis. 

Since $c_2 = 0$ for CNOT equivalence class, in FIG.~\ref{fig3}a, the second point $(e^{2ic_2})$ coincides with first point and the fourth point $(-e^{2ic_2})$ coincides with third point. Similarly for DCNOT equivalence class $(c_2 = \pi/2)$, the second $(e^{2ic_2})$ and fourth $(-e^{2ic_2})$ points coincide with the third and first points respectively. Hence for these two equivalence classes two chords vanish and the reamining four chords pass through origin. However, the average value of square of chord lengths $(8/3)$ is the same for all SPEs and it is twelve times their entangling power. 

Among SPEs, B gate with $c_2 = \pi/4$ is very special as it is the only gate which can generate all two-qubit gates in two applications~\cite{Zhang2004}. It is also the only local equivalence class which is invariant under mirror operation and hence it is represented by the center of the Weyl chamber in the space spanned by local invariants~\cite{Selvan2023}. It can be noted from FIG.~\ref{fig3}b that the convex hull of the square of eigenvalues of the nonlocal part of B gate encloses maximum area in the unit circle. 

\subsection{SPE Circuits}

We consider the single parameter $(\theta)$ circuits shown in FIG.~\ref{fig4}. The circuit shown in FIG.~\ref{fig4}d is generalization of the B gate circuit given in Ref.~\cite{Zhang2004}. Other three circuits are obtained by either changing the order of operations (FIG.~\ref{fig4}c), changing the control and target qubits (FIG.~\ref{fig4}b) or changing both of them (FIG.~\ref{fig4}a). In all four circuits, $\text{R}_x(\theta) = \exp \left(\dfrac{-i\theta \sigma_x}{2} \right)$ and $\text{X} = \sigma_x$. The local invariants of all the four circuits are $G_1 = 0$ and $G_2 = \cos(\theta)$. Comparing with local invariants of SPEs, we get, $\theta = 2c_2$. Thus these circuits are locally equivalent to SPEs. When $\theta=0$, the circuits belong to CNOT equivalence class and when $\theta= \pi$, they belong to Double-CNOT equivalence class. 

\begin{figure}[h]
\begin{tabular}{cc}
\includegraphics[scale=0.35]{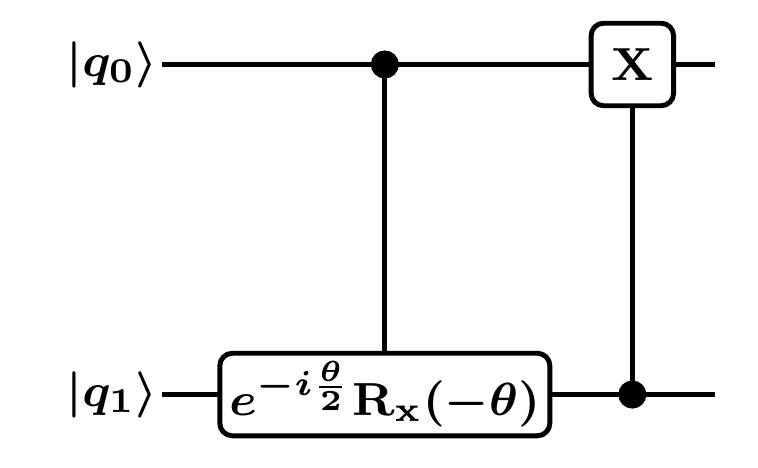} & \includegraphics[scale=0.35]{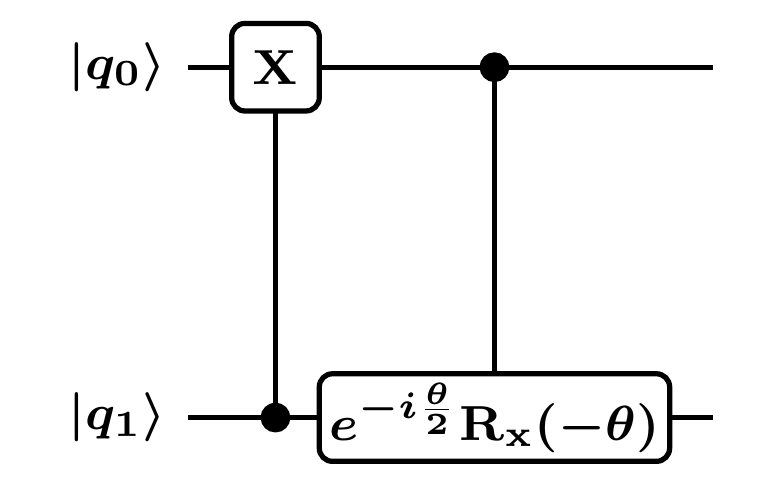} \\ (a) & (b) \\ \includegraphics[scale=0.35]{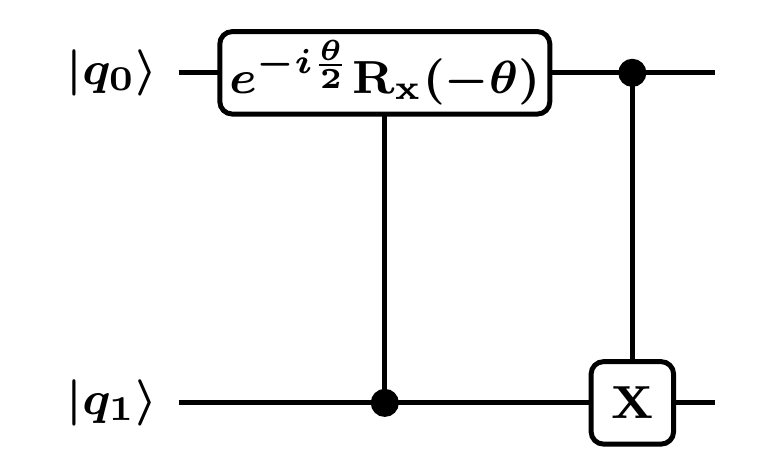} & \includegraphics[scale=0.35]{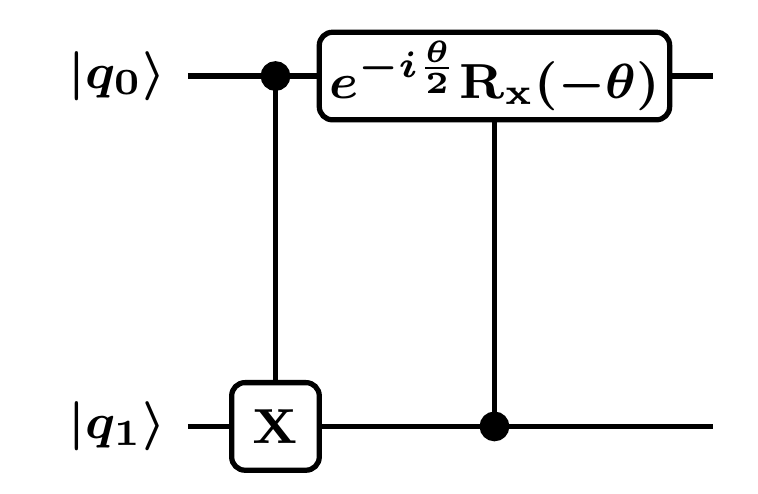} \\ (c) & (d)
\end{tabular}
\caption{Circuits locally equivalent to SPEs.}
\label{fig4}
\end{figure}

In this paper, we write the qubits in the order $\vert q_1 q_0 \rangle$. In this order, the matrix of the circuit shown in FIG.~\ref{fig4}a can be written as follows. 

\begin{equation*}
U_{\mathrm{SPE}}(\theta)= \left[ \begin{array}{c c c c} 1 & 0 & 0 & 0 \\ 0 & 1 & 0 & 0 \\ 0 & 0 & 0 & 1 \\ 0 & 0 & 1 & 0 \end{array} \right]  \times \left[ 
\begin{array}{c c c c}
1 & 0 & 0 & 0 \\ 0 & \frac{1+e^{-i\theta}}{2} & 0 & \frac{1-e^{-i \theta}}{2}  \\ 0 & 0 & 1 & 0 \\ 0 & \frac{1-e^{-i\theta}}{2} & 0 & \frac{1+e^{-i\theta}}{2}
\end{array}
\right]
\end{equation*}
\begin{equation}\label{eq14}
= \left[ 
\begin{array}{c c c c}
1 & 0 & 0 & 0 \\ 0 & \frac{1+e^{-i\theta}}{2} & 0 & \frac{1-e^{-i \theta}}{2} \\ 0 & \frac{1-e^{-i\theta}}{2} & 0 & \frac{1+e^{-i\theta}}{2} \\ 0 & 0 & 1 & 0 
\end{array}
\right].
\end{equation}

The SPE circuit shown in FIG.~\ref{fig4}a transforms the computational product basis $\{ \vert 00 \rangle, \vert 01 \rangle, \vert 10 \rangle, \vert 11 \rangle \}$ as follows. 

\begin{equation*}
\vert 00 \rangle \rightarrow \vert 00 \rangle,
\end{equation*} 
\begin{equation*}
\vert 01 \rangle \rightarrow \left[ \dfrac{1+e^{-i\theta}}{2} \right] \vert 01 \rangle + \left[ \dfrac{1-e^{-i\theta}}{2} \right] \vert 10 \rangle,
\end{equation*}
\begin{equation*}
\vert 10 \rangle \rightarrow \vert 11 \rangle,
\end{equation*} 
\begin{equation}\label{eq15}
\vert 11 \rangle \rightarrow \left[ \dfrac{1-e^{-i\theta}}{2} \right] \vert 01 \rangle + \left[ \dfrac{1+e^{-i\theta}}{2} \right] \vert 10 \rangle.
\end{equation}

When the input states are $\vert 01 \rangle$ and $\vert 11 \rangle$, this circuit generates entangled states for $\theta \in (0, \pi)$. When $\theta=\pi/2$, the generated states are maximally entangled. 

\section{Implementation of SPE circuit}

In the following, we discuss the implementation of SPE circuit shown in FIG.~\ref{fig4}a, in the 127 qubit \textit{ibm\textunderscore kyoto} processor. For all two-qubit experiments, we used qubits 7 and 6 of \textit{ibm\textunderscore kyoto} processor as $q_0$ and $q_1$ respectively. ECR gate is used as basis gate in \textit{ibm\textunderscore kyoto} processor. ECR gate is implemented as a sequence of operations shown below. 

\begin{equation}\label{eq16}
\text{ECR} = e^{i \pi (\sigma_x \otimes \sigma_z)/8 } \left[I \otimes X \right] e^{-i \pi (\sigma_x \otimes \sigma_z)/8 }
\end{equation}

\begin{figure}[h]
\centering
\includegraphics[width=0.5\textwidth]{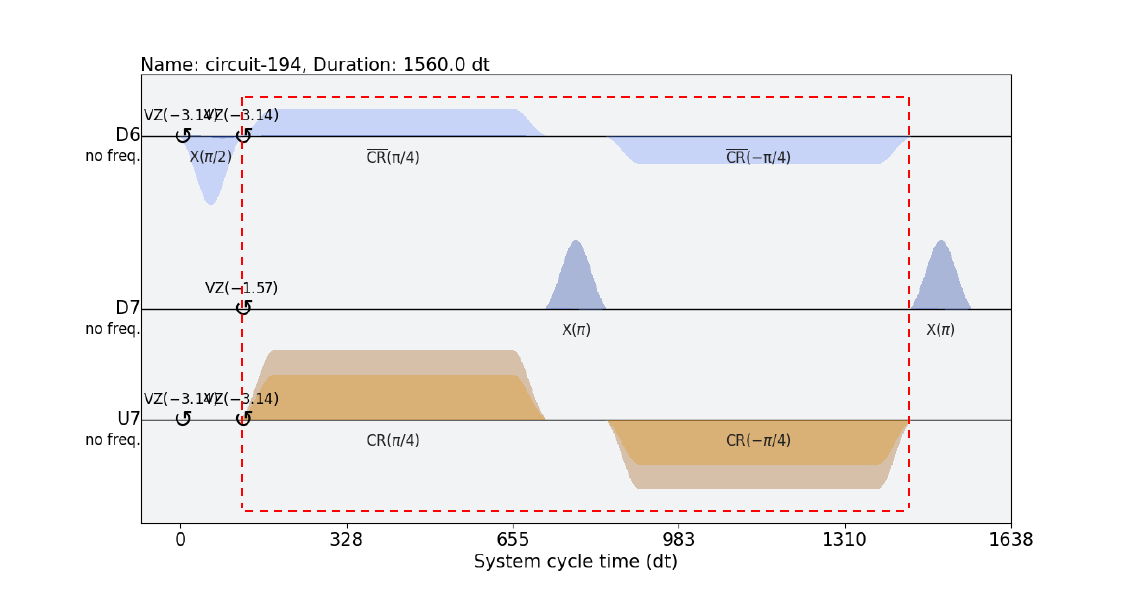}
\caption{Pulse sequence of CNOT gate between qubits 7 (control) and 6 (target) of \textit{ibm\textunderscore kyoto} processor. Pulse sequence enclosed inside red dashed box corresponds to ECR gate.}
\label{fig5}
\end{figure}

This operation involves application of $\pi/4$ cross-resonance (CR) pulse in the control channel (U7), $\pi$ pulse on the control qubit (Drive channel D7) and $-\pi/4$ CR pulse in the control channel (U7) as shown inside the red dashed box in FIG.~\ref{fig5}. In addition to CR pulses, ECR gate involves application of rotary pulses~\cite{Sheldon2016,Sundaresan2020} in the drive channel D6. The entire pulse sequence shown in FIG.~\ref{fig5} corresponds to CNOT gate with a global phase. This is obtained by applying local gates before and after ECR as shown below.

\begin{equation}\label{eq17}
e^{i\pi/4} \text{CNOT} = [I \otimes X] \text{ECR} \left[e^{i\pi \sigma_x/4 } \otimes e^{i \pi \sigma_z/4 } \right].
\end{equation}

Controlled-$\left[e^{-i\theta/2}\text{R}_\text{X}(-\theta)\right]$ operation can be performed with a global phase of $e^{i\theta/4}$ by the sequence of operations $[ I \otimes X ] [ \text{ECR} ]^{\theta/\pi} \left[ e^{i {\theta}\sigma_x /{4} } \otimes e^{i {\theta}\sigma_z/{4} } \right]$ with $\left[ \text{ECR} \right]^{\theta/\pi} = e^{i {\theta}(\sigma_x \otimes \sigma_z)/{8} } \left[ I \otimes X \right] e^{-i {\theta}(\sigma_x \otimes \sigma_z)/{8} }$. 

CR and rotary pulses in the ECR gate have Gaussian square envelopes;  $e^{\pm i {\theta}(\sigma_x \otimes \sigma_z)/{8} }$ in $[ \text{ECR} ]^{\theta/\pi}$ can be implemented by scaling the area of Gaussian square pulses as discussed in Ref.~\cite{Stenger2021,Earnest2021}. For scaling the area, we consider the following Gaussian square waveform. 
\begin{equation*}
f(t) = A f'(t)
\end{equation*}
with 
\begin{equation}\label{eq18}
f'(t) =
\begin{cases}
\exp \bigg[- \dfrac{(t-r)^2}{2\sigma^2}\bigg], & 0 \leq t \leq r \\ & \\
1, & r \leq t \leq r+w \\   & \\
\exp \bigg[- \dfrac{(t-(r+w))^2}{2\sigma^2} \bigg], & r+w \leq t \leq d
\end{cases}
\end{equation}

where $A$ is complex amplitude, $\sigma$ is the standard deviation, $w$, $d$ and $r=\dfrac{d-w}{2}$ are the duration of the embedded square pulse, duration of the entire pulse and risefall duration respectively. The area $(S)$ under this Gaussian square waveform is given by 

\begin{equation}\label{eq19}
S= \vert A \vert \bigg[ w + \sqrt{2 \pi} \sigma \text{erf}\bigg(\dfrac{r}{\sqrt{2} \sigma}\bigg) \bigg].
\end{equation}

The area under the Gaussian square waveform required to implement the operation, $e^{\pm i {\theta}(\sigma_x \otimes \sigma_z)/{8} }$, is given by 

\begin{equation}\label{eq20}
S_\theta = \dfrac{\theta}{\pi} \times S_\pi, 
\end{equation}

where $S_\pi$ is the area corresponding to CNOT operation. The waveform with the area $S_\theta$ can be obtained by changing either $w$ or both $w$ and $\vert A \vert$ as described in Ref.~\cite{Stenger2021}. 

Thus, using pulse-level programming~\cite{Alexander2020}, the SPE circuits can be implemented by scaling the area of CR and rotary pulses of controlled-$\left[e^{-i\theta/2}\text{R}_\text{X}(-\theta)\right]$ operation. Alternatively, the matrices of the SPE circuits [Eq.~\ref{eq14}] can be directly implemented using two ECR gates and local gates. However, implementation using pulse-level programming with rescaled CR and rotary pulses takes lesser implementation time for $\theta \in (0, \pi)$. As an example, the circuit implementing the matrix of B gate circuit (SPE circuit with $\theta = \pi/2$) between qubits 7 and 6 of \textit{ibm\textunderscore kyoto} processor and the corresponding pulse sequence are shown in FIG.~\ref{fig6}a and FIG.~\ref{fig6}b respectively. Pulse sequence of the same B gate circuit obtained by scaling the area of CR and rotary pulses of controlled-$\left[e^{-i\pi/4}\text{R}_\text{X}(-\pi/2)\right]$ operation is shown in FIG.~\ref{fig6}c. It can be noted that the total pulse duration in FIG.~\ref{fig6}c is approximately 0.8256 times the duration of the pulse sequence shown in FIG.~\ref{fig6}b. 

\begin{figure}
\begin{center}
\begin{tabular}{c}
\includegraphics[width=0.5\textwidth]{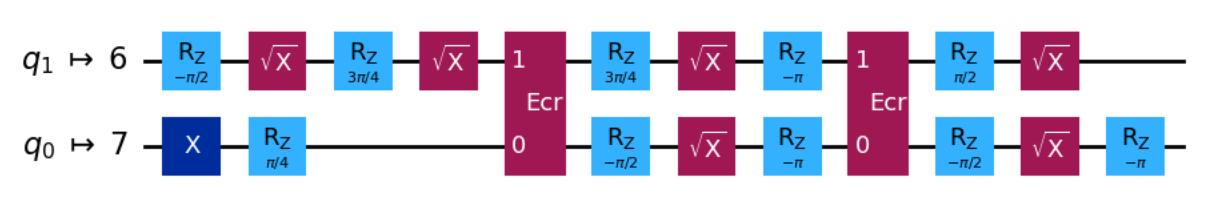} \\ (a) \\
\includegraphics[width=0.5\textwidth]{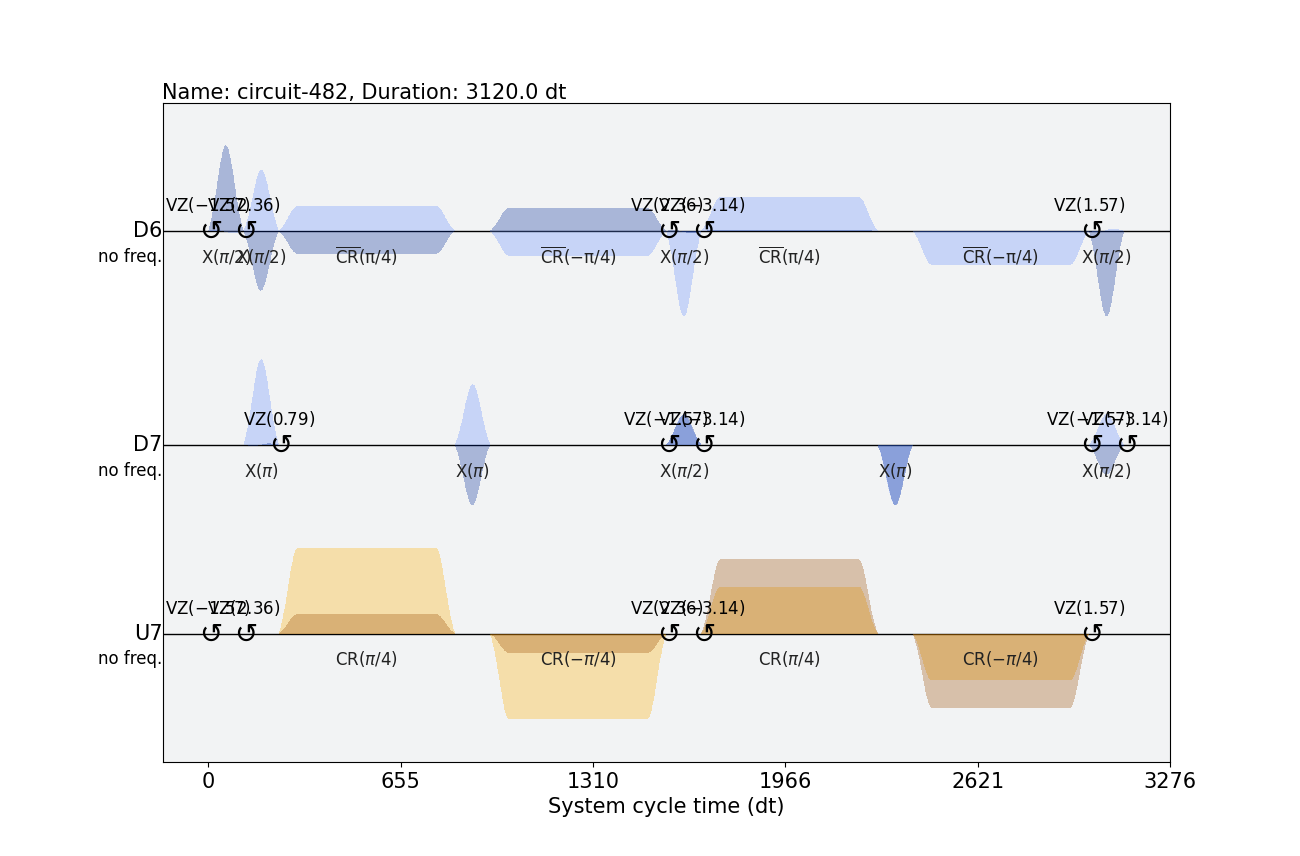} \\ (b) \\
\includegraphics[width=0.5\textwidth]{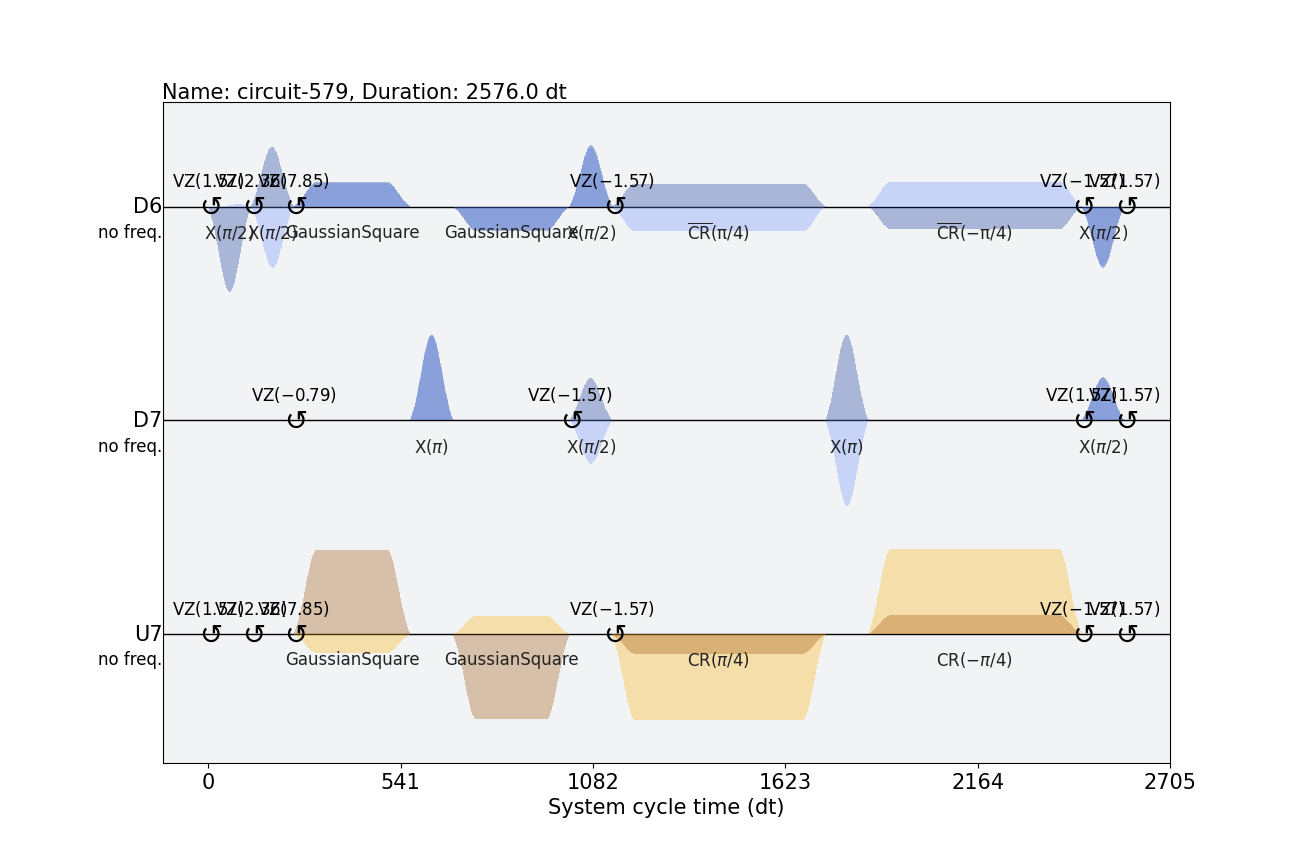}\\ (c)
\end{tabular}
\end{center}
\caption{(a) Circuit implementing the matrix of B gate circuit between qubits 7 and 6 of \textit{ibm\textunderscore kyoto} processor and (b) the corresponding pulse sequence. (c) Pulse sequence obtained by scaling the area of CR and rotary pulses which implements the same B gate circuit.}
\label{fig6}
\end{figure}

We implemented the SPE circuit using two ECR gates and pulse-level programming. The pulse-level programming method is used to rescale the area of CR and rotary pulses to implement controlled-$\left[e^{-i\theta/2}\text{R}_\text{X}(-\theta)\right]$ operation for nine equally spaced values of $\theta$ from $\theta = 0.1 \pi$ to $\theta=0.9 \pi$. We measured the probabilities of $\vert 01 \rangle$ and $\vert 10 \rangle$ for the input states $\vert 01 \rangle$ and $\vert 11 \rangle$ and plotted them as function of $\theta/\pi$ in FIG.~\ref{fig7}. 

\begin{figure}
\begin{tabular}{c}
\includegraphics[scale=0.5]{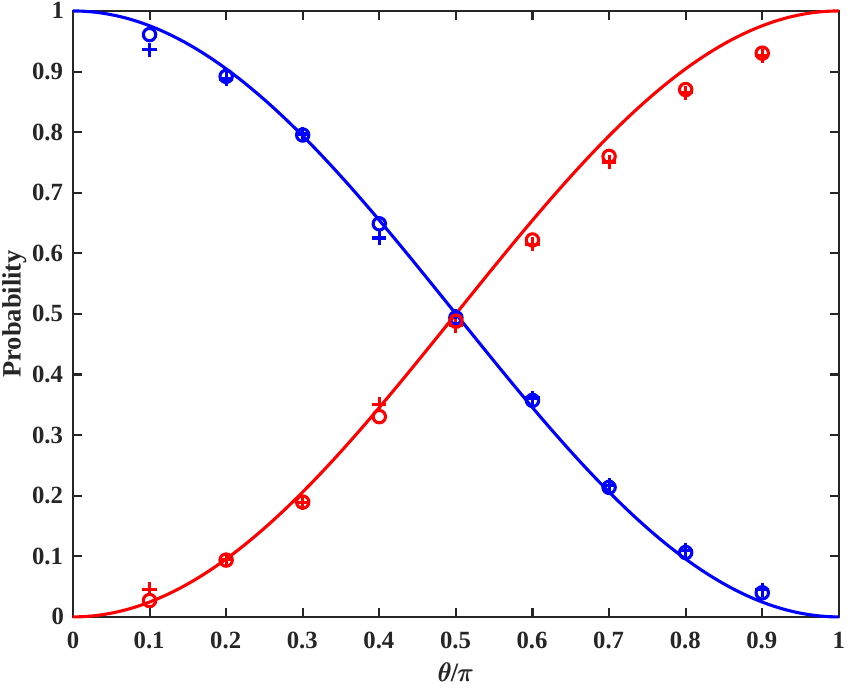} \\ (a) \\
\includegraphics[scale=0.5]{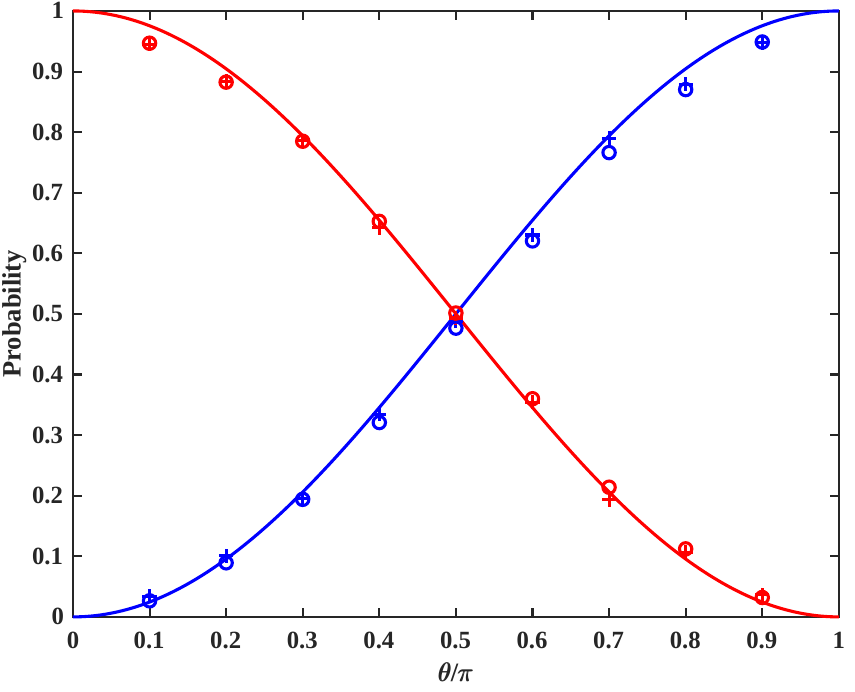} \\ (b)
\end{tabular}
\caption{Probabilities of the states $\vert 01 \rangle$ (blue in color) and $\vert 10 \rangle$ (red in color) for the input states (a) $\vert 01 \rangle$ and (b) $\vert 11 \rangle$. In both subfigures, the solid curves represent theoretical probabilities, the symbols $+$ and $\circ$ represent the probabilities in the case of implementation using two ECR gates and pulse-level programming respectively.}
\label{fig7}
\end{figure}

For the input state, $\vert 11 \rangle$, we performed QST experiment by doing measurements in the bases $\{XX, XY, XZ, YX, YY, YZ, ZX, ZY, ZZ\}$ and constructed the valid output density matrices as described in Ref~\cite{Smolin2012}. We calculated the fidelity of the obtained density matrices using the following formula~\cite{Jozsa1994}. 
 
\begin{equation}\label{eq20a}
F(\rho, \sigma) = \left[tr\left[ \sqrt{\sqrt{\rho} \sigma \sqrt{\rho} }\right] \right]^2
\end{equation}

where $\rho$ and $\sigma$ are the obtained and target density matrices respectively. We have plotted the fidelities in FIG.~\ref{fig7a}a. We also calculated the concurrence~\cite{Wootters2001} of the obtained density matrices and compared it with theoretical result in FIG.~\ref{fig7a}b.

\begin{figure}
\begin{tabular}{c}
\includegraphics[scale=0.5]{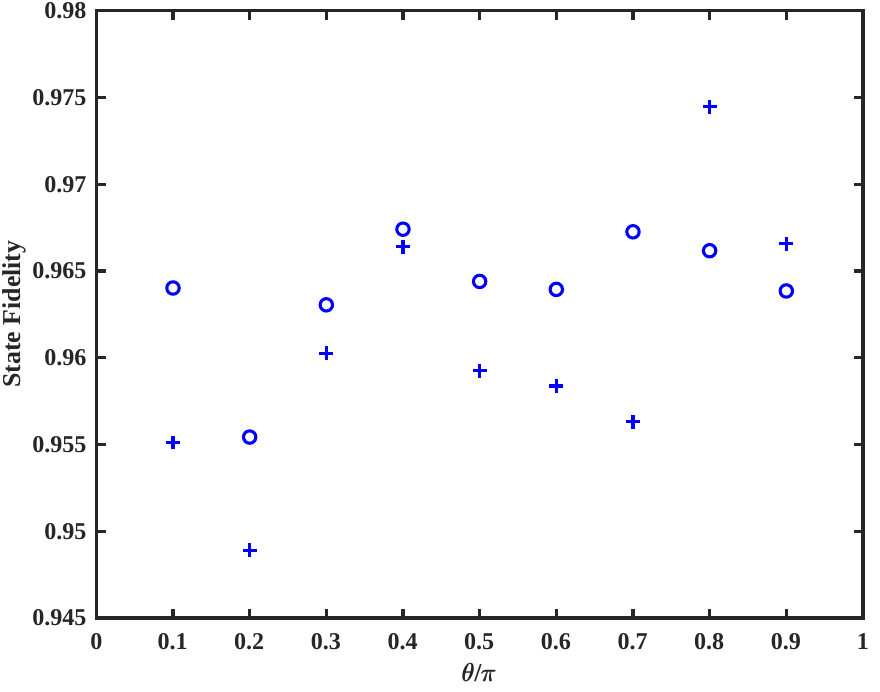} \\ (a) \\
\includegraphics[scale=0.5]{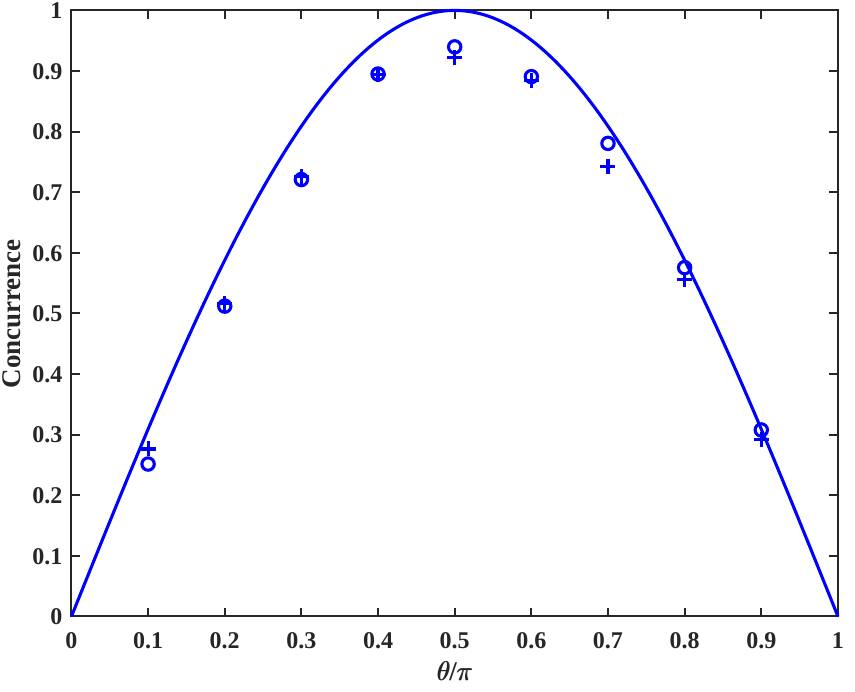} \\ (b)
\end{tabular}
\caption{(a) State fidelity and (b) Concurrence of the output states corresponding to the input state $\vert 11 \rangle$. In both subfigures, the symbols $+$ and $\circ$ correspond to the cases of implementation using two ECR gates and pulse-level programming respectively. The solid curve in the subfigure (b) is theoretical concurrence curve.}
\label{fig7a}
\end{figure}

We measured the average gate fidelity of the B gate circuit. The qubits are prepared in 16 combinations of the states $\left\{ \vert 0 \rangle, \vert 1 \rangle, \left[\vert 0 \rangle + \vert 1 \rangle \right]/\sqrt{2}, \left[\vert 0 \rangle + i \vert 1 \rangle \right]/\sqrt{2} \right\}$. For each of these input states, the output density matrix of the B gate circuit is determined by performing QST experiment. When the output density matrix of the B gate circuit is not a valid density matrix the likelihood density matrix is found using the method described in Ref.~\cite{Smolin2012}. Using these likelihood density matrices, the transformation of basis matrices, $\mathcal{E}(\rho_{k})$'s $(1 \leq k \leq 16)$ [$\rho_k$'s are $4 \times 4$ matrices with elements labelled by numbers from 1 to 16; in $\rho_k$, the $k^{\text{th}}$ element is one and all other elements are zero] are found. The average gate fidelity is found using the following formula~\cite{Nielsen2002}.

\begin{equation}\label{eq21}
\overline{F}(\mathcal{E},U) = \dfrac{16 + \sum_{jk} tr[\rho_k^\dagger U_j] tr[U U_j^{\dagger} U^{\dagger} \mathcal{E}(\rho_k)]}{80}
\end{equation}

where $U$ is the target unitary operation, and $U_j \in \{A \otimes B \}$ with $A, B \in \{I, X, Y, Z \}$.

The average gate fidelity obtained for the B gate circuit implemented by pulse-level programming is 0.96270 and in the case where it is implemented using ECR gate, it is 0.95511. In both cases of implementation, the results shown in FIG.~\ref{fig7}, FIG.~\ref{fig7a} and the average gate fidelity of the B gate circuit are nearly the same.  

\section{Applications}

As applications of the SPE circuits, we discuss the construction of universal two-qubit quantum circuits (UTQQCs) and generation of genuinely entangled multi-qubit states. Construction of UTQQCs using SPEs was reported in Ref.~\cite{Selvan2023}. We construct such UTQQCs using the SPE circuits as shown in FIG.~\ref{fig9}. The parameter of the SPE circuits in FIG.~\ref{fig9} is the Cartan co-ordinate $c_2 \in [0, \pi/2]$ instead of $\theta$. Thus the SPEs involved in the UTQQCs are from CNOT gate to B gate circuit. The parameters of the local gates $\text{R}_\text{Y}(c_{1(3)}) = \exp\left( \dfrac{-ic_{1(3)} \sigma_y}{2} \right)$ are the other two Cartan co-ordinates. The local invariants of the two circuits shown in FIG.~\ref{fig9} can be found as 

\begin{equation}\label{eq22}
G_1 = \dfrac{1}{4} \left[ \sum_{p=1}^3 \cos (2c_p) + \prod_{q=1}^3 \cos (2c_q) + i \prod_{r=1}^3 \sin (2c_r) \right]
\end{equation}
\begin{equation}\label{eq23}
G_2 = \sum_{p=1}^3 \cos (2c_p)
\end{equation}

These are the general expressions of local invariants given in Eqs.~\ref{eq3} and \ref{eq4} in terms of Cartan co-ordinates~\cite{Watts2013}. 

\begin{figure}[h]
\begin{tabular}{c}
\includegraphics[scale=0.6]{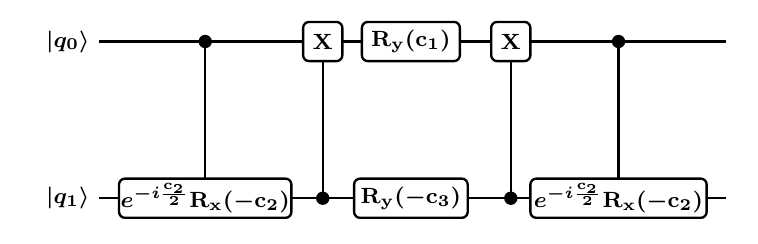} \\ (a) \\
\includegraphics[scale=0.6]{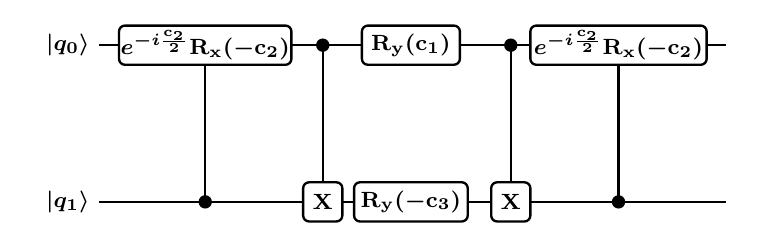} \\ (b)
\end{tabular}
\caption{Universal two-qubit quantum circuits (UTQQCs) constructed using SPE circuits shown in FIG.~\ref{fig4}.}
\label{fig9}
\end{figure}

Thus using these circuits any two-qubit gate belonging to a local equivalence class with Cartan co-ordinates $(c_1, c_2, c_3)$ can be generated by adding suitable single-qubit gates before and after these circuits. These circuits can easily be implemented using rescaled CR and rotary pulses as discussed in the last section.  

From the transformations given in Eq.~\ref{eq15}, it can be verified that the SPE circuit (shown in FIG.~\ref{fig4}a) transform the product basis $\left\{ \dfrac{\vert 00 \rangle \pm \vert 10 \rangle}{\sqrt{2}}, \dfrac{\vert 01 \rangle \pm \vert 11 \rangle}{\sqrt{2}} \right\}$ into the maximally entangled basis $\left\{ \dfrac{\vert 00 \rangle \pm \vert 11 \rangle}{\sqrt{2}}, \dfrac{\vert 01 \rangle + \vert 10 \rangle}{\sqrt{2}}, \dfrac{e^{-i\theta} \left[\vert 01 \rangle - \vert 10 \rangle \right]}{\sqrt{2}}  \right\}$. In the case of multi-qubit systems, there exist two inequivalent classes of genuinely entangled states: GHZ class and W class~\cite{Dur2000}. The GHZ state can be used to perform perfect teleporation of a qubit and superdense coding~\cite{Karlsson1998,Hao2001}. Perfect W states belonging to W class states can also used for perfect teleportation and superdense coding~\cite{Agrawal2006}. The B gate circuit can be used $(n-1)$ times to generate $n$-qubit GHZ and perfect W states as shown in FIG.~\ref{fig10}. 

\begin{figure}[h]
\begin{tabular}{c}
\includegraphics[scale=0.6]{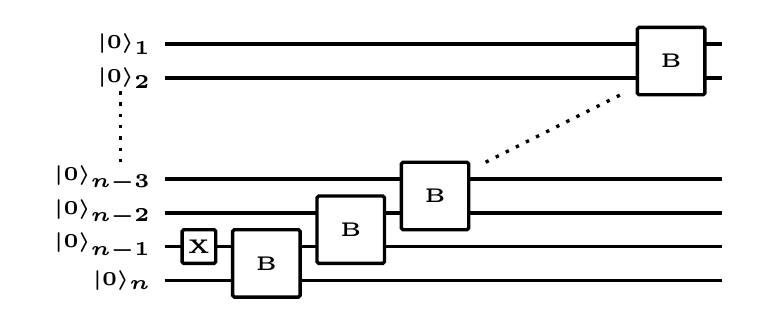} \\ (a) \\
\includegraphics[scale=0.6]{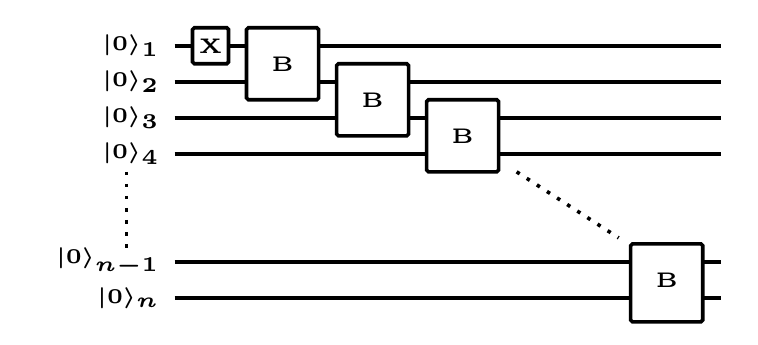} \\ (b)
\end{tabular}
\caption{Circuits consisting of $(n-1)$ B gates to generate (a) $n$-qubit GHZ class state and (b) $n$-qubit perfect W state.}
\label{fig10}
\end{figure}

The matrix of B gate is given by $U_{\text{SPE}}(\pi/2)$ [Eq.~\ref{eq14}]. The ciruit shown in FIG.~\ref{fig10}a transforms the initial product state as shown below.

\begin{equation}\label{eq24}
\bigotimes_{j=0}^{n-1} \vert 0 \rangle_j \rightarrow \dfrac{e^{-i\pi/4}}{\sqrt{2}} \left[ \vert 0 \rangle_{n-1} \bigotimes_{j=0}^{n-2} \vert 1 \rangle_j + i \vert 1 \rangle_{n-1} \bigotimes_{j=0}^{n-2} \vert 0 \rangle_j \right]
\end{equation}

The transformed state belongs to GHZ class which can be converted into GHZ state using suitable single-qubit gates. Similarly, the transformation caused by the circuit shown in in FIG.~\ref{fig10}b is shown below. 

\begin{equation}\label{eq25}
\bigotimes_{j=0}^{n-1} \vert 0 \rangle_j \rightarrow \dfrac{e^{-i\pi/4}}{\sqrt{2}} \left[ \vert 1 \rangle_0 \bigotimes_{j=1}^{n-1} \vert 0 \rangle_j + i \vert \Phi \rangle \right]
\end{equation}

where $\vert \Phi \rangle$ is the normalized $n$-qubit state given as 

\begin{equation*}
\vert \Phi \rangle = \dfrac{e^{i(n-2)\pi/4}}{2^{(n-2)/2}} \vert 1 \rangle_{n-1} \bigotimes_{j=0}^{n-2} \vert 0 \rangle_j +
\end{equation*}
\begin{equation}\label{eq26}
 ~~~~~~~\sum_{k=1}^{n-2} \dfrac{e^{i(n-k-3) \pi/4}}{2^{(n-k-1)/2}} \vert 1 \rangle_{n-k-1} \bigotimes_{j=0,j \neq n-k-1}^{n-1} \vert 0 \rangle_j
\end{equation}

The transformed state is a $n$-qubit perfect W state~\cite{Swain2023}. Similar to B gate, $(n-1)$ applications of ECR or CNOT gate can also generate $n$-qubit GHZ state~\cite{Cruz2019}. But the $n$-qubit perfect W state can't be generated with $(n-1)$ ECR gates. If the B gate circuit is implemented using the circuit shown in FIG.~\ref{fig6}a, then $2(n-1)$ ECR gates are required to generate $n$-qubit perfect W state. Pulse-level programming implementation of B gate takes lesser time than the implementation using ECR gates. Hence pulse-level programming implementation of B gate can be used to generate the perfect W state with better fidelity. To demostrate this, we generated three-qubit perfect W state in \textit{ibm\textunderscore kyoto} processor by using B gate circuit. We used qubits 7, 6, and 5 of \textit{ibm\textunderscore kyoto} processor as $q_0$, $q_1$ and $q_2$ respectively. For three-qubit case, the output state of the circuit shown in FIG.~\ref{fig10}b is given by 

\begin{equation}\label{eq27}
\vert \text{W}_\text{P} \rangle_3 = \dfrac{e^{-i\pi/4}}{\sqrt{2}} \left[ \vert 001 \rangle + i \dfrac{e^{-i\pi/4}}{\sqrt{2}} \vert 010 \rangle + i \dfrac{e^{i\pi/4}}{\sqrt{2}} \vert 100 \rangle \right]
\end{equation}

\begin{figure}[h]
\begin{tabular}{c}
\includegraphics[scale=0.425]{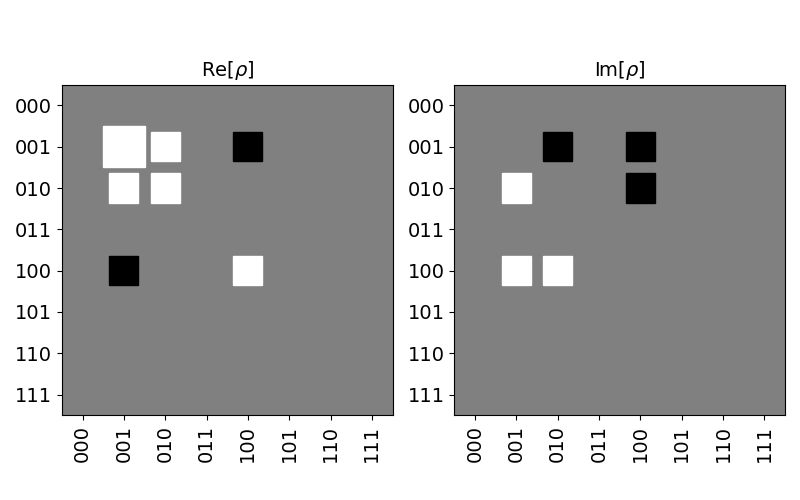} \\ (a) \\
\includegraphics[scale=0.425]{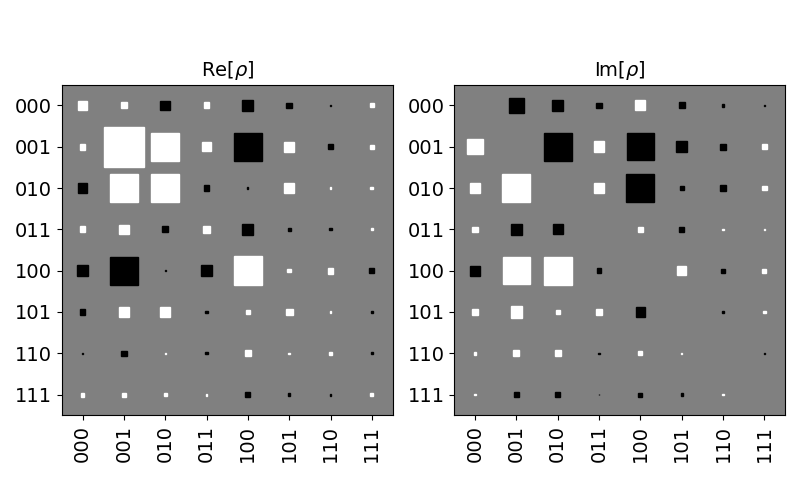} \\ (b) \\
\includegraphics[scale=0.425]{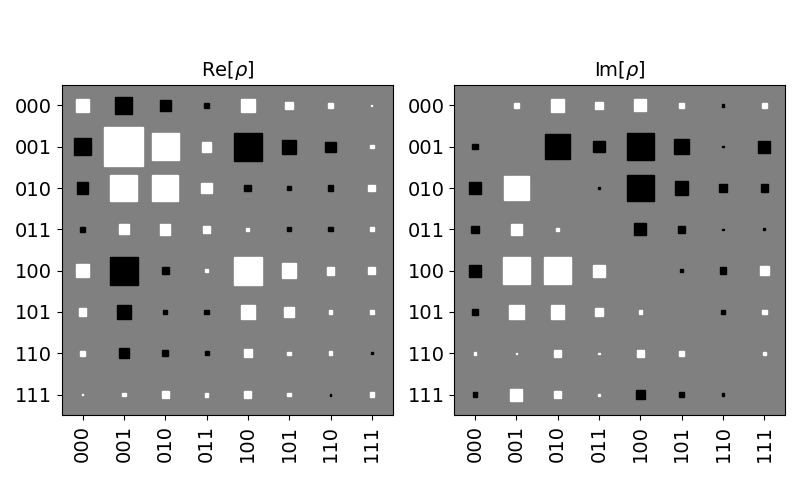} \\ (c)
\end{tabular}
\caption{Hinton diagram of the real and imaginary parts of the density matrix of three-qubit perfect W state given in Eq.~\ref{eq27}. (a) ideal case, (b) density matrix of the state generated using pulse-level programming, and (c) density matrix of the state generated using four ECR gates.}
\label{fig11}
\end{figure}

 We implemented the B gate circuit using pulse-level programming and ECR gates. We performed QST experiment, found the output density matrix and displayed them in FIG.~\ref{fig11}. We found the fidelity using Eq.~\ref{eq20a}. When the B gate circuit is implemented using rescaled pulses the fidelity of the obtained density matrix is $F=0.91753$. When it is implemented using ECR gates the obtained fidelity is $F=0.85432$. The state generation using the rescaled pulse implementation of B gate has provided better fidelity. 

\section{Conclusion}

To summarize, we have discussed the characteristics of argand diagram of squared eigenvalues of nonlocal part of SPEs, implementation and applications of SPE circuits with a special focus on B gate circuit. We have shown that the chord lengths in argand diagram of squared eigenvalues of nonlocal part of a two-qubit gate can be used to describe the ability of the two-qubit gate to generate entanglement. For each chord in the argand diagram associated with a two-qubit gate, there exists a pair of orthonormal product states constructed from the states belonging to magic basis such that the amount of entanglement generated out of those product states is proportional to the chord length. From the results presented in Ref.~\cite{Kraus2001}, we can say that the maximum entanglement generated by a non-perfect entangling two-qubit gate is proportional to the length of longest chord in the argand diagram associated with it. For the two-qubit gates that can transform a pair of orthonormal product state into a pair of orthonormal maximally entangled states, a chord passes through origin. For SPEs, two chords connecting two disjoint pairs of points pass through origin. Hence, SPEs can transform a product basis into maximally entangled basis. We have shown that the entangling power is proportional to the average of squared chord lengths. This result implies that the entangling power of a two-qubit gate depends only on the average entanglement it can generate over the product states $\left\{ \vert \psi^\pm_{jk} \rangle ; ~jk \in [12, 13, 14, 23, 24, 34] \right\}$. 

We have implemented SPE circuit using both ECR gate and pulse-level programming and presented the results. The average gate fidelities obtained for the B gate circuit implemented using ECR gate and pulse-level programming are nearly the same. However, when the B gate has to be used more than once, pulse-level programming is the more suitable method of implementation. Because, it reduces the total implementation time and improves the fidelity of operation. This is evident from the QST experiment results obtained for the generation of three-qubit perfect W state. The matrices corresponding to the four B gate circuits [FIG.~\ref{fig4} for $\theta=\pi/2$] have simple form and are more suitable to perform calculations. Hence, in the future, if the nonlocal part of B gate is used as native gate in IBM quantum processors~\cite{Wei2023}, then the matrix corresponding to one of the four B gate circuits can be used as basis gate similar to CNOT gate. 

Contruction of two-qubit gates using CX and fractional CX gates is discussed in Ref.~\cite{Peterson2020}. UTQQCs constructed using SPE circuits simplify the task of finding the decomposition of an arbitrary two-qubit gate. Hence, the UTQQCs will be very useful for doing quantum computation in IBM processors provided the cross-resonance interaction can be achieved accurately for arbitrary values of $\theta$ and the qubits have long relaxation and dephasing times. 
 
\vspace*{1em}

\textbf{Acknowledgement:} We acknowledge the use of IBM Quantum services for this work. The views expressed are those of the authors, and do not reflect the official policy or position of IBM or the IBM Quantum team.




\end{document}